\documentclass[reqno,12pt,letterpaper]{amsart}
\PassOptionsToPackage{usenames,dvipsnames}{xcolor}
\usepackage{amsmath,amssymb,amsthm,graphicx,mathrsfs,url,bbm,cases,enumitem,tikz-cd}
\usepackage[colorlinks=true,linkcolor=Red,citecolor=Green]{hyperref}
\usepackage{amsxtra}
\usepackage[toc,page]{appendix}
\usepackage{wasysym} 
\usepackage{graphicx}
\usepackage{subcaption}
\usepackage{comment}
\def\arXiv#1{\href{http://arxiv.org/abs/#1}{arXiv:#1}}

\usepackage{soul}

\usepackage{mathtools}

\def\?[#1]{\textbf{[#1]}\marginpar{\Large{\textbf{??}}}}

\def\smallsection#1{\smallskip\noindent\textbf{#1}.}
\let\epsilon=\varepsilon 

\usepackage{CJKutf8}
\newcommand{\RR}{{\mathbb R}}
\newcommand{\NN}{{\mathbb N}}
\newcommand{\CC}{{\mathbb C}}
\newcommand{\TT}{{\mathbb T}}
\newcommand{\ZZ}{{\mathbb Z}}

\newtheorem{theo}{Theorem}
\newtheorem{prop}{Proposition}[section]	
\newtheorem{defi}[prop]{Definition}

\newtheorem{qu}{Question}

\newtheorem*{assumption*}{Assumption}
\newtheorem{lemm}[prop]{Lemma}
\newtheorem{corr}[prop]{Corollary}
\newtheorem{rem}{Remark}

\numberwithin{equation}{section}

\DeclareMathOperator{\loc}{loc}

\DeclareMathOperator{\supp}{supp}

\DeclareMathOperator{\tr}{tr}

\usepackage{scalerel}

\newcommand\reallywidehat[1]{\arraycolsep=0pt\relax%
\begin{array}{c}
\stretchto{
  \scaleto{
    \scalerel*[\widthof{\ensuremath{#1}}]{\kern-.5pt\bigwedge\kern-.5pt}
    {\rule[-\textheight/2]{1ex}{\textheight}} 
  }{\textheight} %
}{0.5ex}\\           
#1\\                 
\rule{-1ex}{0ex}
\end{array}
}
\author[S.\,Becker]{Simon Becker}
\address[Simon Becker]{ETH Zurich, 
Institute for Mathematical Research, 
R\"amistrasse 101, 8092 Zurich, 
Switzerland}
\email{simon.becker@math.ethz.ch}

\author[Z.\,Tao]{Zhongkai Tao}
\address[Zhongkai Tao]{Department of Mathematics, University of California, Berkeley, CA 94720, USA}
\email{ztao@math.berkeley.edu}

\author[M.\,Yang]{Mengxuan Yang}
\address[Mengxuan Yang]{Department of Mathematics, University of California, Berkeley, CA 94720, USA}
\email{mxyang@math.berkeley.edu}

\title[Wannier decay]{Wannier decay and the Thouless conjecture}

\begin{document}
\begin{abstract}
Non-trivial Chern classes pose an obstruction to the existence of exponentially decaying Wannier functions which provide natural bases for spectral subspaces. For non-trivial Bloch bundles, we obtain decay rates of Wannier functions in dimensions $d=2,3$. For $d=2$, we construct Wannier functions with full asymptotics and optimal decay rate $\mathcal{O}(|x|^{-2})$ as conjectured by Thouless; for $d=3$, we construct Wannier functions with the uniform decay rate $\mathcal{O}(|x|^{-7/3})$.
\end{abstract}
\maketitle

\section{Introduction}
\label{sec:intro}
For periodic systems in quantum mechanics, topological phases are characterized by Chern classes of the corresponding bundles of Bloch eigenfunctions. Vanishing of Chern classes is equivalent to the existence of exponentially decaying Wannier functions (cf.~Brouder~et~al. \cite{br} and Monaco~et~al.~\cite{MPPT18}). 
It is natural to consider the optimal decay rate of Wannier functions for periodic systems with non-trivial topology. For topologically non-trivial two-dimensional periodic systems, Thouless \cite{T84} conjectured that there exist Wannier functions with the precise decay rate $\mathcal{O}(|x|^{-2})$ as $\vert x \vert \to \infty$.

In this paper, we prove Thouless' conjecture by explicitly constructing Wannier functions with full asymptotics and decay rate $\mathcal{O}(|x|^{-2})$. In addition, for $d=3$, we extend our method to obtain Wannier functions with two different decay rates: a uniform decay rate $\mathcal{O}(|x|^{-7/3})$ and an inhomogeneous decay rate $\mathcal{O}((1+|x_1|+|x_2|)^{-2} |x_3|^{-\infty})$. 

\subsection{Backgrounds and main results}
We start by introducing some general concepts.
\begin{assumption*}
\label{ass:proj}
Let $\Gamma$ be a lattice in $\mathbb R^d$ and $\mathcal P:=(P(k))_{k \in \mathbb R^d}$  a family of orthogonal projections with finite rank $r$ acting on $\mathcal H = L^2(\mathbb R^d / \Gamma)$. We assume  real analytic dependence on $k$,   $ P(k)\mathcal{H} \subset C^{\infty} (\mathbb R^d / \Gamma)$ and that for all $k \in \mathbb R^d$ and $\gamma \in \Gamma^*$
\begin{equation}
\tag{A}
\label{eq:t-equivariance}
P(k+\gamma) =\tau(\gamma)^{-1}P(k)\tau(\gamma), \quad \text{ where $ \tau (\gamma) $
is a unitary operator}. 
\end{equation}
\end{assumption*}
Perhaps the most important examples of $P(k)$ are spectral projections to $r$ isolated bands of periodic Hamiltonians with smooth coefficients and $\tau(\gamma):=e^{-i\langle \gamma,x \rangle}$. 

Projections as in Assumption \eqref{eq:t-equivariance} induce a vector bundle via the equivalence relation
\[ (k,\varphi) \sim (k',\varphi') \Leftrightarrow (k',\varphi')=(k+\gamma,\tau(\gamma)\varphi),\]
with the total space \[E:=\{ (k,\varphi) \in \mathbb R^d \times \operatorname{ran}(P(k))\} /{\sim} \] and the base space $M:=\mathbb R^d/\Gamma^*\cong \TT^d$. Thus the projection map $\pi: E\to M$ induces a complex vector bundle $\mathcal E$, which is called the \textbf{Bloch bundle}. 

A \textbf{Wannier basis} of the Bloch bundle $\mathcal{E}$ is a set of orthonormal functions
\begin{equation}
    \label{eq:wb}
     \mathcal W := \bigcup_{a=1}^r \bigcup_{\gamma\in\Gamma} \big\{ \psi_a ( x - \gamma ) \big\} 
\end{equation}
such that $\mathcal W$ is an orthonormal basis of the range of the projection $\mathcal{P}$ with rank $r$. The topology of a Bloch bundle $\mathcal{E}$ for $d = 2,3$ is characterized by the first Chern class $c_1(\mathcal{E})$. It is well known that the existence of exponentially decaying Wannier functions is equivalent to the topological triviality ($c_1(\mathcal{E})=0$) of the Bloch bundle \cite{ne,P07,MPPT18}. See Section \ref{sec:vb} for a detailed discussion.  

For non-trivial Bloch bundles, we consider the decay rate of functions $\{\psi_a(x)\}$ in a Wannier basis $\mathcal{W}$. Our first result resolves Thouless' conjecture by explicitly constructing Wannier functions with full asymptotics in dimension two:
\begin{theo}
\label{theo:thouless}
    Let $\mathcal E$ be a Bloch bundle of rank $r$ with Chern number $c_1(\mathcal{E}) = m \neq 0$ over $\mathbb T^2$, then there exists a Wannier basis $\mathcal{W} = \{\psi_a(x-\gamma): \gamma\in\Gamma, a\in\{1,\cdots, r\}\}$ such that $\psi_a(x)$ decays exponentially for $a \in\{1,\cdots, r-1\}$  and 
    \begin{equation*}
     \psi_r(x) \sim \frac{1}{|\RR^2/\Gamma^*|}\sum\limits_{\alpha\in\NN^2} 
    c_{m,\alpha}(\theta_x)
     \vert x \vert^{-2-|\alpha|}\partial_{k}^{\alpha}\Phi(0,x)
    \end{equation*}
    with the leading coefficient $c_{m,0}(\theta_x)=2\pi (-1)^m me^{-im \theta_x}$,
    $\sin\theta_x=x_1/|x|$, $\cos \theta_x =x_2/|x|$, and $\Phi(k,x)$ is a normalized smooth local section of $\mathcal E$ near $k=0$.
\end{theo}

Our second result gives a Wannier decay rate for non-trivial Bloch bundles over $\TT^3$:
\begin{theo}
\label{theo:becker}
    Let $\mathcal E$ be a non-trivial Bloch bundle of rank $r$ over $\TT^3$, then there exists a Wannier basis $\mathcal{W} = \{\psi_a(x-\gamma): \gamma\in\Gamma, a\in\{1,\cdots, r\}\}$ such that $\psi_a(x)$ decays exponentially for $a \in\{1,\cdots, r-1\}$  and 
    \[\psi_r(x) = \mathcal O(\vert x \vert^{-\frac{7}{3}}), \text{ as }\vert x \vert \to \infty.\]
\end{theo} 

\subsection{Wannier decay and regularities of Bloch frames} 
The decay rate of Wannier functions is closely related to the regularity of the corresponding  global Bloch frames of the Bloch bundles $\mathcal{E}$, from which the Wannier functions are constructed. We introduce 
\begin{defi}
\label{def:bf}
For the orthogonal projections $\mathcal P:=(P(k))_{k \in \mathbb R^d}$ satisfying Assumption \eqref{eq:t-equivariance}, we say that $\Phi:\mathbb R^d \to \mathcal H^r$ is a \emph{global Bloch frame}, if 
   \[\Phi : \RR^d \to \mathcal{H} \oplus \dots \oplus \mathcal{H} = \mathcal{H}^r,\quad k \mapsto (\varphi_1(k), \dots, \varphi_r(k)),\]
   is \emph{$\tau$-equivariant}:
    \begin{equation}
        \label{eq:equivariance}
        \varphi_a(k + \gamma) = \tau(\gamma) \varphi_a(k) \quad \text{for all } k \in \mathbb{R}^d, \, \gamma \in \Gamma^*, \, a \in \{1, \dots, r\},
    \end{equation}
and for a.e.\ $k \in \RR^d$, the set $\{\varphi_1(k), \dots, \varphi_r(k)\}$ is an orthonormal basis spanning $\operatorname{Ran} P(k)$. 
\end{defi}

\begin{rem}
\label{rem:section}
    In the language of vector bundles, each $\varphi_a(k)$ is a normalized section of the Bloch bundle, and a Bloch frame $\Phi$ is a family of orthonormal sections of the Bloch bundle $\mathcal{E}$ that spans the Bloch bundle over $\CC$ for a.e. $k\in \RR^d$. 
\end{rem}

\begin{defi}
We say that the global Bloch frame $\Phi = (\varphi_a)_{1\leq a\leq r}$ in Definition \ref{def:bf} is
\begin{enumerate}
    \item \emph{continuous} (respectively \emph{smooth, analytic}) if the maps $\varphi_a : \mathbb{R}^d \to \mathcal{H}$ are continuous (respectively smooth, analytic) for all $a \in \{1, \dots, r\}$;
    \item $H^s$-\emph{regular} if the maps $\varphi_a : \mathbb{R}^d \to \mathcal{H}$ lie in the corresponding local Sobolev space $H^s_{\loc}(\mathbb{R}^d; \mathcal{H})$ for all $i \in \{1, \dots, r\}$.
\end{enumerate}
\end{defi}
We now define Wannier functions for a given Bloch frame $\Phi$. First, recall the \textbf{Bloch transform} is given by 
 \begin{equation}
     \label{eq:BFT}
     (\mathcal B\psi)(k,x) :=\sum_{\gamma \in \Gamma} e^{-i\langle x+\gamma,k\rangle} \psi(x+\gamma), \quad \psi\in L^2(\mathbb R^d), 
 \end{equation}
which is an isometry 
 \[\mathcal B: L^2(\mathbb R^d) \longrightarrow \mathcal{H}_\tau := \{\varphi \in L^2_{\text{loc}}(\RR^d;L^2(\mathbb R^d/\Gamma )): \varphi(k+\gamma,x)=\tau(\gamma)\varphi(k,x) \text{ for all }\gamma \in \Gamma^*\}. \]
 
\begin{defi}
\label{def:Wannier}
Let $\varphi(k,x) \in \mathcal H_{\tau}$ be a normalized section of the Bloch bundle, the \emph{Wannier function} $w(\varphi) \in L^2(\mathbb R^d)$ is defined by
\begin{equation}
\label{eq:Wannier}
 w(\varphi)(x):=\frac{1}{\vert \mathbb R^d / \Gamma^* \vert} \int_{\mathbb R^d / \Gamma^*} e^{i\langle x, k\rangle}\varphi(k,x) \ dk=\mathcal B^{-1}(\varphi). 
\end{equation}
Moreover, Wannier functions of the Bloch frame $\Phi = (\varphi_a)_{1\leq a\leq r}$ are given by $\{w(\varphi_a)\}_{1\leq a\leq r}$.
\end{defi}

For a global Bloch frame $\Phi = (\varphi_a)_{1\leq a\leq r}$, the shifted Wannier functions
\begin{equation}
\label{eq:wb2}
\mathcal{W} = \{w(\varphi_a)(\bullet - \gamma)\}_{\gamma \in \Gamma, 1\leq a\leq r}   
\end{equation}
form an orthonormal basis of the range of the projection $\mathcal{P}$ with rank $r$. 
Hence, the set $\mathcal{W}$ is indeed a \emph{Wannier basis}. See \cite[Section~6.5]{ku2} or \cite[Proposition~5.5]{notes} for detailed discussions. 

For any Bloch bundle with base dimension $d\leq 3$, one can always find a global Bloch frame $\Phi = (\varphi_a)_{1\leq a\leq r}$ being $H^s$-regular for all $s<1$ (cf.~\cite[Theorem~2.4]{MPPT18}). Hence, under the inverse Bloch transform \eqref{eq:Wannier}, the corresponding Wannier functions $w(\varphi_a)$ satisfy the weighted $L^2$-estimate
\begin{equation}
    \label{eq:Hs-condition}
    \sum_{a=1}^r \Vert \langle \bullet \rangle^s w(\varphi_a)\Vert_{L^2(\mathbb R^d)}<\infty\text{ for all }s<1.
\end{equation}
Such $H^s$-regular Bloch frame (with $s<1$) can be improved to an $H^1$-regular or an analytic Bloch frame if and only if the Bloch bundle has vanishing Chern classes \cite{MPPT18}, which then yields the existence of exponentially decaying Wannier functions $\{w(\varphi_a)\}_{1\leq a\leq r}$.

In terms of pointwise decay rates, for non-trivial Bloch bundles, Thouless \cite{T84} conjectured that Wannier functions should be able to attain some optimal decay rate $\mathcal{O}(|x|^{-2})$ for $\vert x \vert \gg 1$ when $d=2$. His argument rests on the Bloch function representation \eqref{eq:Wannier}:
If $w(\varphi) = \mathcal O(\langle \bullet \rangle^{-d-\varepsilon})$ for $\varepsilon>0,$ then the series converges uniformly 
\[\vert \varphi(k,x)\vert \le \sum_{\gamma \in \Gamma} \vert w(\varphi)(x+\gamma)\vert \lesssim \sum_{\gamma \in \Gamma} \langle x+\gamma \rangle^{-d-\varepsilon} <\infty. \]
This shows that $k \mapsto \varphi(k,x) $ is a continuous normalized global section of the Bloch bundle. Hence the Bloch bundle is trivial. This leads Thouless to conjecture that the pointwise decay rate $\mathcal O(|x|^{-2})$ may be achieved by some Wannier functions for $d=2$ (also see \cite[Theorem~5.4]{ku}). 

It was pointed out in \cite{MPPT18} that the weighted $L^2$-condition \eqref{eq:Hs-condition} heuristically corresponds to the pointwise decay rate $\mathcal O(|x|^{-2})$ for
$d = 2$ and $\mathcal O(|x|^{-5/2})$ for $d = 3$. However, the existence of Wannier functions saturating such decay rate remains open. 

On the other hand, it is also worthwhile to point out that the pointwise decay rate $w(\varphi) = \mathcal{O}(|x|^{-2})$ is stronger than the weighted $L^2$-condition \eqref{eq:Hs-condition} for $d=2$, in the sense that it satisfies a better weighted $L^2$-estimate
\[\Vert \langle \bullet \rangle^1 (\log \langle \bullet \rangle)^{-t} w(\varphi)\Vert_{L^2(\mathbb R^d)}<\infty\text{ for all } t>1/2.\]
For $d=2$, \cite{LDLK24} construct a distinguished section of the Bloch bundle that matches the decay $\mathcal{O}(|x|^{-2})$. Our Theorem \ref{theo:thouless} gives a different construction and refines \cite{LDLK24} by showing that all Wannier functions may decay exponentially except one that exhibits the Thouless decay rate with a complete asymptotic expansion. 

Interestingly, for $d=3$, Thouless' simple continuity argument fails to capture a sharp uniform decay rate. Our construction can be carried over to Bloch bundles over $\TT^3$ (see Section \ref{sec:t3}) to construct Bloch frames in $H^s$ for all $s<1$ with the corresponding Wannier function decay rate $\mathcal{O}((1+|x_1|+|x_2|)^{-2}\langle x_3\rangle^{-\infty})$ as $x\to\infty$ (see Corollary \ref{cor:Becker}). 
In addition, the uniform decay rate we obtained in Theorem \ref{theo:becker} is better than the decay rate for Wannier functions on $\TT^2$ in all directions.

It is reasonable to ask whether a better uniform decay rate in all directions is possible. Since the function $\langle x \rangle^{-\frac{5}{2}}$ satisfies the weighted $L^2$-bound $\Vert \langle \bullet \rangle w \Vert_{L^2(\mathbb R^3)}=\infty$ and $\Vert \langle \bullet \rangle^s w \Vert_{L^2(\mathbb R^3)}<\infty$ for $s<1$, we propose
\begin{qu}
For a non-trivial Bloch bundle over $\mathbb T^3$, what is the optimal pointwise decay rate of the Wannier functions as $|x| \to\infty$? Can one construct a Wannier function $w(\varphi)$ such that $w(\varphi)(x)=\mathcal{O}(|x|^{-\frac{5}{2}})$ as $\vert x\vert\to\infty$? 
\end{qu}

Moreover, we propose the following more general question for vector bundles over any smooth manifold.
\begin{qu}
    Given a smooth vector bundle $\mathcal{E}$ of rank $r$ over a smooth manifold $M$. What is the largest $\alpha\in \RR$ such that there exists a (measurable) normalized frame $s_i$, $i=1,2,\cdots,r$ with respect to a smooth metric, such that in any local trivialization over an open subset $U\subset M$ and for any smooth cutoff $\chi\in C_c^{\infty}(U)$, we have
    \begin{equation*}
        \mathcal{F}(\chi s_i)(\xi) = \mathcal{O}(|\xi|^{-\alpha})
    \end{equation*}
    as $|\xi|\to \infty$?
\end{qu}
The only general thing we can say is that $\alpha$ can be arbitrarily large when $\mathcal{E}$ is trivial, and $\alpha\leq d:=\dim M$ when $\mathcal{E}$ is nontrivial. As we see above, there could be finer restrictions on $\alpha$, and it is related to the sharp Fourier decay under certain topological constraints.

\smallsection{Other related works}
The classical question on the existence of exponentially localized Wannier functions can be fully understood in terms of Chern classes, including the works by Nenciu \cite{ne}, Marzari--Vanderbilt \cite{MV}, Panati \cite{P07}, Brouder--Panati--Calandra--Mourougane--Marzari \cite{br}, Marzari--Mostofi--Yates--Souza--Vanderbilt \cite{M+12} and Monaco--Panati--Pisante--Teufel \cite{MPPT18}.  
For periodic systems with non-trivial Chern classes, even though it is not possible to construct orthonormal Bloch frames that lead to exponentially decaying Wannier basis, it has been shown that one can construct so-called Parseval frames with redundancies by embedding it into trivial Bloch bundles with higher ranks (see the works by Kuchment \cite{ku,ku2}, Auckly--Kuchment \cite{ak} and Cornean--Monaco--Moscolari \cite{cmm}). This yields a set of exponentially localized Wannier functions that are \emph{linearly dependent}. Recently, results on Wannier function localization have been extended to non-periodic systems by Lu--Stubbs--Watson \cite{LSW}, Marcelli--Moscolari--Panati \cite{MMP}, Lu--Stubbs \cite{LS24} and Rossi--Panati \cite{RP24}. 

We would also like to point out that, for dimension two, our construction of normalized sections of Bloch bundles is similar to \cite[Theorem~5.1]{MPPT18} but from a different perspective, while the full asymptotics of Wannier functions is new as far as we know. For dimension three, the construction of the normalized section in Proposition \ref{prop:vanish-t3} and the resulting Wannier decay $\mathcal{O}(|x|^{-7/3})$ are new to the best of our knowledge. 
 
\smallsection{Notations and conventions} Let $\mathcal H$ be a separable Hilbert space. Let $\Gamma = \sum_{i=1}^d \ZZ e_i \subset \mathbb R^d,$ for linearly independent $e_i,$ be a lattice, and $\Gamma^*:=\{ x\in \mathbb R^d : \langle x,\gamma\rangle \in 2\pi \mathbb Z \text{ for all } \gamma \in \Gamma\}=\sum_{i=1}^d \ZZ v_i$ be the dual lattice. We denote by $U(\mathcal H)$ the group of unitary operators on $\mathcal H$ and by $L(\mathcal H)$ the space of bounded linear operators on $\mathcal H$. 

The lattices $\Gamma$ and $\Gamma^*$ naturally give rise to a torus as well as a dual torus, that is, $\mathbb R^d/\Gamma$ and $\mathbb R^d/\Gamma^*$. However, for simplicity, we will often just say that we have a vector bundle over a torus $\mathbb T^d$, which then in the physics setting corresponds to the dual torus $\mathbb R^d /\Gamma^*$.

\smallsection{Structure of the paper}
In Section \ref{sec:vb}, we recall the classification of complex vector bundles using Chern classes.  
In Section \ref{sec:t2}, we compute the full asymptotics and the optimal decay rate of Wannier functions on $\TT^2$ for a given non-trivial Chern class and prove Theorem \ref{theo:thouless}. In Section \ref{sec:t3}, we discuss the decay rate of Wannier functions on $\TT^3$ for non-trivial Chern classes and prove Theorem \ref{theo:becker}. 

\smallsection{Acknowledgment} We would like to thank Maciej Zworski for many helpful discussions and suggestions on the manuscript, and Ruixiang Zhang for helpful discussions on the Fourier decay of singular sections. ZT and MY were partially supported by the Simons Targeted Grant Award No.~896630. MY acknowledges the support of AMS-Simons travel grant. SB acknowledges support by SNF Grant PZ00P2 216019.

\section{Complex vector bundles and Wannier functions}
\label{sec:vb}

In this section, we first discuss the classification of complex vector bundles $\mathcal{E}$ over a manifold of dimension $d\leq 3$ and recall the basic properties of Chern classes.

\subsection{Classification of complex vector bundles} 
We first briefly recall the following
\begin{defi}
\label{def:vb}
    Let $E, M$ be topological spaces. $\pi: E\to M$ is called a complex \emph{vector bundle} of rank $r$ if for any $x\in M$, $\pi^{-1}(x)$ is a complex vector space of dimension $r$, and there exists a covering $\{U_i\}$ of $M$ such that there is a homeomorphism, called the trivialization, which is linear on each fiber $\pi^{-1}(x)$, such that the following diagram commutes 
    \[
    \begin{tikzcd}
    \pi^{-1}(U_i) \arrow[r, "\cong"] \arrow[d, "\pi"] & U_i \times \mathbb{C}^r \arrow[ld, "pr_1"] \\
    U_i &
    \end{tikzcd}
    \]
    Here $E$ is also called the \emph{total space} and $M$ is called the \emph{base}. A vector bundle of rank $1$ is called a line bundle.
\end{defi}
\begin{rem}\label{rem:analytic}
     In the above definition, if $E,M$ are $C^s$ manifolds, $\pi$ is $C^s$ and the trivialization is a $C^s$ map, then the vector bundle is called a $C^s$ vector bundle, where $s\in \NN$ or $s=\infty$ (smooth) or $s=\omega$ (real analytic). We note that the classifications of $C^s$ vector bundles are equivalent for $s\in\NN$ or $s=\infty$ or $s=\omega$, see \cite[Theorem 5]{S64}.

\end{rem}

Now we recall the classification of complex vector bundles following \cite{bty}. For a complex vector bundle $\mathcal{E}$ over a manifold $M$, the Chern class $c_k(\mathcal{E})$ is an element of $H^{2k}(M;\ZZ)$.
If the vector bundle has a smooth connection with curvature $\Omega$ which is a $\mathfrak{gl}(r,\CC)$-valued $2$-form, the Chern classes can be computed by
\begin{equation*}
    \det\left(I+\frac{\sqrt{-1}}{2\pi}t[\Omega]\right) =\sum_{j\geq 0} c_j(\mathcal{E}) t^j.
\end{equation*}
Over any manifold $M$, complex rank $r$ line bundles are classified up to isomorphisms by the first Chern class $c_1 \in H^2(M;\ZZ)$ (see \cite[p.~34]{C79}). For higher-rank complex vector bundles, we recall the following classification result:
\begin{prop}
\label{prop:chern}
    Complex vector bundles $\mathcal{E}$ over a manifold $M$ of dimension $\leq 3$ are classified by the first Chern class $c_1 \in H^2(M;\ZZ)$. Furthermore, $\mathcal{E}$ admits a decomposition into line bundles 
    \begin{equation}\label{eq:complex-bundle-decomp}
        \mathcal{E} = \bigg(\bigoplus_{i=1}^{r-1} L_i\bigg) \oplus L,
    \end{equation}
    where all line bundles $L_i$ are trivial apart from possibly $L$ for which $c_1(\mathcal{E})=c_1(L).$
\end{prop}

\begin{proof}
Let $\mathcal{E}$ be a complex vector bundle of rank $r \geq 2$ over $M$. By the Thom transversality theorem, there exists a non-vanishing section $s$ of $\mathcal{E}$. If $s$ intersects the zero section $0_M$ transversally at $(x,0)\in\mathcal{E}$, we have 
\[ ds (T_xM) + T_{(x,0)}0_M = T_{(x,0)}\mathcal{E}.\]
Thus, the dimension of the intersection of the tangent space
\[
\dim_{\RR}(ds (T_xM) \cap T_{(x,0)}0_M) = \dim_{\RR}(ds (T_xM)) + \dim_{\RR}(T_{(x,0)}0_M) - \dim_{\RR}( T_{(x,0)}\mathcal{E}) < 0,
\]
ensures that the intersection is empty. Hence, $s$ induces a trivial line bundle $L_1$, and $\mathcal{E}$ splits as a Whitney sum:
\[
\mathcal{E} \cong L_1 \oplus \mathcal{E}',
\]
where $\mathcal{E}'$ is a complex vector bundle of rank $r-1$. By induction, $\mathcal{E}$ is isomorphic to the Whitney sum of $r-1$ trivial line bundles with a complex line bundle $L$. The Chern number of $L$ satisfies $c_1(L) = c_1(\mathcal{E})$, which completes the proof.
\end{proof} 
\begin{rem}
    When $\mathcal{E}$ is real analytic, the decomposition in the previous Proposition can be chosen so that $L_i$ and $L$ are real analytic line bundles. This is because \eqref{eq:complex-bundle-decomp} holds in the topological category, and Remark~\ref{rem:analytic} implies a real analytic isomorphism
    \begin{equation*}
         \mathcal{E} \cong \bigg(\bigoplus_{i=1}^{r-1} L_i\bigg) \oplus L.
    \end{equation*}
\end{rem}

\subsection{Chern numbers and localization dichotomy}
\label{sec:Chern_review}
Let $\Omega(k) = \sum_{i,j} \Omega_{ij}(k) \ dk_i \wedge dk_j$ be the curvature form of the Berry connection:
\[ \Omega_{ij}(k):=\tr_{\mathcal H}(P(k)[\partial_i P(k),\partial_j P(k)]).\]
Here, $P(k)$ denotes the family of orthogonal projections parametrized by $k \in \mathbb{R}^d$ that satisfy the equivariance conditions \eqref{eq:t-equivariance} of the reciprocal lattice $\Gamma^*$. Then, the first Chern class is defined as
\[ c_1(\mathcal E)=\tfrac{i}{2\pi} [\Omega] \in H^2_{\text{dR}}(\mathbb R^d/\Gamma^*;\RR),\]
where $[\Omega]$ represents the de Rham cohomology class of the curvature form $\Omega$.

To compute the first Chern class explicitly, we use the de Rham isomorphism. For the $d$-dimensional torus $\mathbb{R}^d / \Gamma^*$, the second homology group satisfies
\[
H_2(\mathbb R^d/\Gamma^*) \cong \mathbb Z^{\binom{d}{2}}.
\]
This group is generated by the independent $2$-cycles $(B_{ij})_{1\le i < j \le d},$ where $B_{ij}:=(\mathbb R v_i + \mathbb R v_j)/(\ZZ v_i + \ZZ v_j)$ and $(v_i)_{i=1}^d$ are basis vectors of $\Gamma^*$.

\begin{defi}
    A family of projections $\mathcal P$ that satisfy Assumption~\eqref{eq:t-equivariance} is called \emph{Chern trivial} if for $d \in \{2,3\}$ the Chern numbers
    \[ C_1(\mathcal E)_{ij}:=\frac{i}{2\pi} \int_{B_{ij}} \tr_{\mathcal H}(P(k)[\partial_i P(k),\partial_j P(k)]) \ dk_i \wedge dk_j\]
    vanish for all $1\le i<j\le d$. 
\end{defi}

By the de Rham isomorphism, Chern triviality implies the vanishing of the first Chern class $c_1(\mathcal E)=0.$ The first Chern class has direct implications on the structure of Wannier functions. This is the \emph{localization dichotomy} as stated, for instance, in \cite{MPPT18}:
\begin{itemize}
    \item \textbf{Case 1:} If $\mathcal P$ is Chern-trivial, then there exists a Bloch frame $\Phi$ that is real analytic in the parameter $k$, such that all Wannier functions $w(\varphi_i)$ decay exponentially.
    \item \textbf{Case 2:} If $\mathcal P$ is not Chern-trivial, then there does not exist Bloch frame that is $H^1_{\loc}$ in $k$. In addition, for any Bloch frame $\Phi$, there exist Wannier functions $w(\varphi_i)$ for some $i$ such that $ x w(\varphi_i) \notin {L^2} $.
\end{itemize}

\section{Asymptotics of Wannier functions on $\TT^2$}
\label{sec:t2}
In this section, we compute the asymptotics of Wannier functions when the Chern number is nonzero for a complex line bundle $\mathcal{E}$ over $\TT^2$. By the (proof of) Proposition \ref{prop:chern} that a general complex vector bundle can be decomposed into the Whitney sum of $r-1$ trivial line bundles $L_i$ and a complex line bundle $L$ such that $c_1(L)=c_1(\mathcal{E})$. 
Trivial bundles $L_i$ always give rise to exponentially localized Wannier functions. Thus, the optimal Wannier decay of a general complex vector bundle over $\TT^2$ or $\TT^3$ is reduced to the optimal Wannier decay rate of a single complex line bundle, which we will investigate next. 

The asymptotics of Wannier functions is related to possible singularities of the normalized section. We first prove the following proposition.

\begin{prop}
\label{prop:vanish}
Let \(\mathcal{E}\) be a smooth complex line bundle over \(\mathbb{T}^2\). If the Chern number \(c_1(\mathcal{E}) = m \in \mathbb{Z}\), then there exists a smooth section \(s: \mathbb{T}^2 \to \mathcal{E}\) that vanishes at a single point and is of the form $z^{m}$ in local coordinates.
\end{prop}

\begin{proof}

Pick a point \(p \in \mathbb{T}^2\). Let \(D \subset \mathbb{T}^2\) be a small disc around \(p\) where \(\mathcal{E}\) is trivialized, so we can identify the bundle locally with \(D \times \mathbb{C}\). In this trivialization, define a section \(s_{\text{local}}\) by
$s_{\text{local}}(z) = z^m$, 
where \(z\) is a complex coordinate on \(D\) with \(z = 0\) corresponding to \(p\). This section has a zero of order \(m\) at \(p\).

Now extend the domain of consideration to a slightly larger disc \(D' \supset D\) with boundary \(\partial D'\). On the complement \(\mathbb{T}^2 \setminus D'\), the line bundle \(\mathcal{E}\) is trivial because \(\mathbb{T}^2 \setminus D'\) is homotopically equivalent to \(\mathbb{S}^1 \vee \mathbb S^1\), and every complex line bundle over \(\mathbb{S}^1 \vee \mathbb S^1\) is trivial as $H^2(\mathbb{S}^1 \vee \mathbb S^1,\ZZ)=0$. Therefore, there exists a smooth, non-vanishing section \(s_{\text{outer}}\) defined on \(\mathbb{T}^2 \setminus D'\) (see also \cite[Proposition 11.14]{BT}). Note that in particular the winding number of the section $s_{\text{outer}}$ on \(\partial D'\) equals the Chern number of $\mathcal{E}$ by \cite[Theorem 11.16]{BT}, as $s_{\text{outer}}$ can be viewed as a section of a $S^1$-bundle since it is non-vanishing.


Now we have a section $s_{\text{local}}$ on $\overline{D}$ and a nonvanishing section $s_{\text{outer}}$ on $\mathbb{T}^2\setminus D'$ such that their winding number on the boundary agrees. We want to glue them smoothly on $D'\setminus D$. Since they have the same winding number, there is a homotopy $H(t,\theta):[0,1]\times \mathbb{S}^1\to \CC\setminus \{0\}$ such that $H(0,\theta)=s|_{\partial D}$ and $H(1,\theta)=s'|_{\partial D'}$.
Therefore, we can define a continuous gluing by
\begin{equation*}
    \widetilde{s}(r,\theta)= H\left(\frac{r-r_D}{r_{D'}-r_D},\theta\right)\in\CC\backslash\{0\},
\end{equation*}
where \(r_D, r_{D'}\) are the radii of \(D\) and \(D'\), respectively. 
By convolving with a smooth approximation of identity near the gluing region, we can make it smooth and still nonvanishing. 


By construction, the section \(s\) vanishes only at \(p\), where it coincides with \(s_{\text{local}}(z) = z^m\). Away from \(p\), \(s\) is smooth and non-vanishing, completing the construction.
\end{proof}

Now we compute the asymptotics of Wannier functions using the section constructed in Proposition \ref{prop:vanish}.

\begin{proof}[Proof of Theorem \ref{theo:thouless}] 
We first normalize the section. The normalized section is smooth everywhere, aside from the singularity at the origin given by $\frac{k^m}{|k|^m}\Phi(k,x)$ on a disc $D$, where $\Phi(k,x)$ is a local normalized smooth section. Here we think of $k$ as a complex variable $k=k_1+ik_2$.
Let $\chi\in C_c^{\infty}([0,\infty))$ be a smooth cutoff function such that $\chi\equiv 1$ near zero. We now use \eqref{eq:Wannier} to recover the Wannier function from the Bloch function and integrate over $\mathbb R^2$, due to the presence of the cutoff function. Modulo an $\mathcal O(|x|^{-\infty})$-error, which is given by taking the Bloch transform of the smooth part corresponding to the cutoff function $1-\chi$ away from the singularity, the Wannier function is given by
\begin{equation}
\label{eq:wannier2}
\begin{split}
    w(\varphi)(x)& =\frac{1}{|\RR^2/\Gamma^*|}
    \int_{\RR^2}e^{i\langle k,x\rangle}\frac{k^m} {|k|^m}\chi(|k|)\Phi(k,x)\, dk + \mathcal O(| x |^{-\infty})\\
    & = \frac{1}{|\RR^2/\Gamma^*|}\int_{0}^{2\pi}\int_{0}^{\infty} e^{im\theta}e^{i\rho (x_1\cos\theta+x_2\sin\theta)}\chi(\rho) \Phi(k,x)\rho \,d\rho \, d\theta + \mathcal O ( | x |^{-\infty}) \\
    & = \frac{\Phi(0,x)}{|\RR^2/\Gamma^*|}\int_{0}^{2\pi}\int_{0}^{\infty}  e^{im\theta}e^{i\rho (x_1\cos\theta+x_2\sin\theta)}\chi(\rho) \rho \,d\rho \, d\theta \\
    &\quad +\frac{1}{|\RR^2/\Gamma^*|}   \int_{\RR^2}e^{i\langle k,x\rangle}\frac{k^m} {|k|^m}\chi(|k|)( \Phi(k,x)-\Phi(0,x)) \,dk + \mathcal O(|x|^{-\infty}).
\end{split}
\end{equation} 
The first term in the right-hand side of \eqref{eq:wannier2} is given by
\begin{equation*}
    \begin{aligned}
    &\frac{\Phi(0,x)}{|\RR^2/\Gamma^*|}\int_{0}^{2\pi}\int_{0}^{\infty}  e^{im\theta}e^{i\rho (x_1\cos\theta+x_2\sin\theta)}\chi(\rho) \rho \,d\rho \, d\theta\\
     =&\frac{\Phi(0,x)}{|\RR^2/\Gamma^*|}\int_{0}^{2\pi}\int_{0}^{\infty} e^{im\theta}e^{i\rho |x|\sin(\theta+\theta_x)}\chi(\rho)\rho \, d\rho \, d\theta    \\
    =&\frac{\Phi(0,x)e^{-im\theta_x}}{|\RR^2/\Gamma^*|}\int_{0}^{2\pi}\int_{0}^{\infty}  e^{im\theta}e^{i\rho |x|\sin\theta}\chi(\rho)\rho \, d\rho \, d\theta 
    \end{aligned}
\end{equation*}
where $\sin\theta_x=x_1/|x|$ and $\cos\theta_x=x_2/|x|$.

We can compute the asymptotics as follows.
\begin{equation}\label{eq:computation}
    \begin{split}
    &\int_{0}^{2\pi}\int_{0}^{\infty}  e^{im\theta}e^{i\rho |x|\sin\theta}\chi(\rho)\rho \, d\rho \, d\theta =|x|^{-2}\int_{0}^{2\pi}\int_{0}^{\infty}  e^{im\theta}e^{i\rho\sin\theta}\chi(\tfrac{\rho}{|x|})\rho \, d\rho \, d\theta\\
    =& |x|^{-2}\int_{0}^{2\pi}\int_{0}^{\infty}  e^{im\theta}e^{i\rho\sin\theta}\rho \, d\rho \, d\theta -|x|^{-2}\int_{0}^{2\pi}\int_{0}^{\infty}  e^{im\theta}e^{i\rho\sin\theta}(1-\chi(\tfrac{\rho}{|x|}))\rho \, d\rho \, d\theta\\
    =&|x|^{-2}\int_{0}^{2\pi}\int_{0}^{\infty}  e^{im\theta}e^{i\rho\sin\theta}\rho \, d\rho \, d\theta +\mathcal{O}(|x|^{-\infty}).
    \end{split}
\end{equation}
The second term is $\mathcal{O}(|x|^{-\infty})$ by integration by parts in $\rho$ and $\theta$. Here the integral does not converge in the usual sense of Riemann integral. They are defined in the distributional sense:
\begin{equation*}
    \int_0^{\infty}e^{i\rho \sin\theta}\rho d\rho :=\lim\limits_{\epsilon\to 0+}    \int_0^{\infty}e^{i\rho (\sin\theta+i\epsilon)}\rho d\rho = - \frac{1}{(\sin\theta+i0)^2}.
\end{equation*}
By integration by parts,
\begin{equation*}
    -\int_0^{2\pi} e^{im\theta}(\sin\theta+i0)^{-2}d\theta =-im\int_0^{2\pi}e^{im\theta}\frac{\cos\theta}{\sin\theta+i0}d\theta.
\end{equation*}
Since $\tfrac{\cos\theta}{\sin\theta+i0}=p.v. \cot\theta-i\pi\delta_0(\theta)+i\pi\delta_\pi(\theta)$, we have
\begin{equation*}
    \begin{split}
    &\int_0^{2\pi}e^{im\theta}\frac{\cos\theta}{\sin\theta+i0}d\theta =\int_0^{2\pi}e^{im\theta} p.v. \cot\theta d\theta+i\pi((-1)^m-1) \\
=&  i\int_0^{2\pi}\frac{\sin m\theta \cos\theta}{\sin\theta} d\theta+i\pi((-1)^m-1) =\pi i((-1)^m+1)+\pi i((-1)^m-1) =2\pi i(-1)^m.
    \end{split}
\end{equation*}
Therefore
\begin{equation*}
    \int_{0}^{2\pi}\int_{0}^{\infty}  e^{im\theta}e^{i\rho\sin\theta}\rho \, d\rho \, d\theta = 2\pi m(-1)^m.
\end{equation*}
Another way to compute the integral is to use the Bessel function
\begin{equation*}
    J_{m}(\rho)=\frac{1}{2\pi} \int_0^{2\pi} e^{-im\theta +i\rho \sin\theta } \ d\theta
\end{equation*}
so that
\begin{equation*}
    \int_{0}^{2\pi}\int_{0}^{\infty}  e^{im\theta}e^{i\rho\sin\theta}\rho \, d\rho \, d\theta = 2\pi (-1)^m\int_0^{\infty} J_m(\rho) \rho d\rho.
\end{equation*}

The integral of Bessel function can be computed by differentiating the following Fourier integral for $\xi \in [0,1)$  (see ~\cite[\href{http://dlmf.nist.gov/10.22.E59}{(10.22.59)}]{DLMF})
\[ \int_0^{\infty} J_m(x)e^{i x \xi} \ dx = \frac{e^{im \arcsin(\xi)}}{(1-\xi^2)^{1/2}}.\] 

In conclusion, the first term in the right-hand side of \eqref{eq:wannier2} is given by
\[\frac{\Phi(0,x)}{|\RR^2/\Gamma^*|}\int_{0}^{2\pi}\int_{0}^{\infty}  e^{im\theta}e^{i\rho (x_1\cos\theta+x_2\sin\theta)}\chi(\rho) \rho \,d\rho \ d\theta = \frac{2\pi (-1)^m me^{-im \theta_x}}{|\RR^2/\Gamma^*|\vert x \vert^2}\Phi(0,x) + \mathcal O(\vert x \vert^{-\infty}).\] 

For the second term in the right-hand side of \eqref{eq:wannier2}, we recall 
\[\Phi(k,x)=\Phi(0,x)+\sum\limits_{0<|\alpha|\leq N_0}\frac{k^{\alpha}}{\alpha!}\partial_k^{\alpha}\Phi(0,x)+\sum\limits_{|\beta|=N_0+1}k^{\beta}R_{\beta}(k,x)\]
where $R_{\beta}(k,x)$ is smooth and $k^{\alpha} := k_1^{\alpha_1} k_2^{\alpha_2}$ for $k=k_1+ik_2$. 
By a similar estimate as above, we have
\begin{equation*}
     \int_{\RR^2}e^{i\langle k,x\rangle}\frac{k^m} {|k|^m}k^{\alpha}\chi(|k|) dk =c_{m,\alpha}(\theta_x)|x|^{-2-|\alpha|}+\mathcal{O}(|x|^{-\infty}).
\end{equation*}
Since $\frac{k^m}{|k|^m}k^{\beta}\chi(|k|)\in C^{N_0}$ for $|\beta|=N_0+1$, we have
\begin{equation*}
    \sum\limits_{|\beta|=N_0+1}\int_{\RR^2}e^{i\langle k,x\rangle}\frac{k^m} {|k|^m}k^{\beta}R_{\beta}(k,x)\chi(|k|) dk =\mathcal{O}(|x|^{-N_0}).
\end{equation*}
We conclude from \eqref{eq:wannier2} that
\begin{equation*}
     w(\varphi)(x) \sim \frac{1}{|\RR^2/\Gamma^*|}\sum\limits_{\alpha\in\NN^2} c_{m,\alpha}(\theta_x)\vert x \vert^{-2-|\alpha|}\partial_{k}^{\alpha}\Phi(0,x)
\end{equation*}
with $c_{m,0}(\theta_x)=2\pi (-1)^m me^{-im \theta_x}$.
\end{proof}

\section{Decay of Wannier function on $\TT^3$} 
\label{sec:t3}
In this section, we consider decay of Wannier functions of non-trivial Bloch bundles over $\TT^3$. For simplicity, we assume $\Gamma^*=\ZZ^3$ in this section. The first Chern class for complex vector bundles over $\mathbb T^3$ is given by 
\[ c_1 = m_1 [dx_1 \wedge dx_2] + m_2 [dx_2 \wedge dx_3] + m_3 [dx_3 \wedge dx_1],\quad m_1,m_2,m_3\in\ZZ.\]
To further simplify the presentation, we want to find a change of coordinates on $\mathbb T^3$ to simply the form of $c_1$, which corresponds to a transformation $y=Bx$ with a unimodular matrix $B \in \operatorname{SL}(3,\mathbb Z)$.
We have the following Lemma: 
\begin{lemm}
\label{lem:Oltman-Stubbs}
    There exist $B \in \operatorname{SL}(3,\mathbb Z)$ and new coordinates $y:=Bx$ with $x \in \mathbb R^3 / \mathbb Z^3$ such that the first Chern class is given by \[ c_1=m [dy_1 \wedge dy_2] \text{ where }m := \operatorname{gcd}(m_1,m_2,m_3).\]
\end{lemm}
\begin{proof}
We set
\[ v_1 := \frac{1}{m}(m_1,m_2,m_3) \in \mathbb Z^3.\]
Since $v_1$ is a primitive vector, we can build a unimodular matrix $A := (v_1,v_2,v_3) \in \operatorname{SL}(3,\mathbb Z),$ see  \cite[Lemma $1$]{S60}.

Thus, we have $A e_1 = v_1 $  and $(m_1,m_2,m_3)A^{-T}=(m,0,0)$. 
We want to find coordinates $y:=Bx$ with $B\in \mathrm{SL}(3,\ZZ)$ such that
\[\begin{pmatrix} dx_1 \wedge dx_2, \ dx_2 \wedge dx_3, \ dx_3 \wedge dx_1 \end{pmatrix}^T= A^{-T} \begin{pmatrix}dy_1 \wedge dy_2,  \ dy_2 \wedge dy_3 , \ dy_3 \wedge dy_1\end{pmatrix}^T.\]
Since
\[ dy_i \wedge dy_j =\left(\sum_{k} B_{ik} dx_k\right) \wedge \left(\sum_{l} B_{jl} dx_l\right) = \sum_{k<l} \det\begin{pmatrix}
    B_{ik} & B_{il} \\ B_{jk} & B_{jl} \end{pmatrix} dx_{k} \wedge dx_l,\]
and the adjugate matrix $\mathrm{adj}(B)=B^{-1}$, we have
\[\begin{pmatrix}dy_1 \wedge dy_2,  \ dy_2 \wedge dy_3 , \ dy_3 \wedge dy_1\end{pmatrix}^T=B^{-1}
\begin{pmatrix} dx_1 \wedge dx_2, \ dx_2 \wedge dx_3, \ dx_3 \wedge dx_1 \end{pmatrix}^T. \]
We take $B=A^{-T}$. The Chern class in the $y$ coordinate is $c_1 = m[dy_1\wedge dy_2].$
\end{proof}

In the following we shall assume that the Chern class is of the form described in Lemma \ref{lem:Oltman-Stubbs}. Recall that by Proposition \ref{prop:chern}, complex vector bundles over $\TT^3$ are also classified by the first Chern class. We now extend the construction of Wannier functions in Section \ref{sec:t2} to non-trivial Bloch bundles over $\mathbb T^3$. As in Section \ref{sec:t2}, we will make the construction on a line bundle $\mathcal{E}$ with the first Chern class $(m,0,0)$. We first construct a section of the complex line bundle $\mathcal{E}$ over $\TT^3$, using Proposition \ref{prop:vanish}:

\begin{prop}
\label{prop:vanish-t3}
Let \(\mathcal{E}\) be a smooth complex line bundle over $\mathbb{T}^3=\RR^3/\ZZ^3$. If the Chern number is given by $c_1(\mathcal{E}) = (m,0,0) \in \mathbb{Z}^3$, then there exists a smooth section \(s: \mathbb{T}^3 \to \mathcal{E}\) such that the section only vanishes on a curve $\gamma(t)=(\gamma_1(t),\gamma_2(t),t)=(a\cos(2\pi t), a\sin(2\pi t),t)$ for some sufficiently small $a>0$. Moreover,
\begin{equation*}
    s(k_1,k_2,t)=(k_1+ik_2-\gamma_1-i\gamma_2(t))^m
\end{equation*}
near the curve $\gamma(t)$.
\end{prop}


\begin{proof}
Since the line bundle $\mathcal{E}$ has Chern number $(m,0,0)$, it is isomorphic to the pullback bundle $\pi^*L_m$ where $\pi:\RR^3/\ZZ^3\to \RR^2/\ZZ^2$ is the projection onto the first two coordinates and $L_m$ is the line bundle with Chern number $m$ on $\RR^2/\ZZ^2$. By Proposition~\ref{prop:vanish}, we can find a section of $L_m$ that only vanishes at $(0,0)$ and has the form $k^m$ near $(0,0)$. By pulling back this section to $\RR^3/\ZZ^3$, we get a section $s_{\mathrm{outer}}$ of $\mathcal{E}$ that only vanishes on the curve $(0,0,t)$, $t\in\RR/\ZZ$ and has the form $k^m$ near $(0,0,t)$. Note $\mathcal{E}$ restricted to a tubular neighbourhood $D:=\{(k_1,k_2,t): |(k_1,k_2)|\leq 2c_0, t\in\RR/\ZZ\}$ of the helix is trivial.
Thus we construct a local section of the form
\begin{equation*}
    s_{\mathrm{local}}(k_1,k_2,t):=(k_1+ik_2-\gamma_1(t)-i\gamma_2(t))^m,\quad (k_1,k_2,t)\in D.
\end{equation*}
We now take
\begin{equation*}
    s(k_1,k_2,t)=
    \begin{cases}
        s_{\mathrm{local}}(k_1,k_2,t),\ &|k|\leq c_0 \\
        \frac{2c_0-|k|}{c_0}s_{\mathrm{local}}(k_1,k_2,t)+\frac{|k|-c_0}{c_0}s_{\mathrm{outer}}(k_1,k_2,t),\ & c_0 \leq |k|\leq 2c_0\\
        s_{\mathrm{outer}}(k_1,k_2,t),\ & (k_1,k_2,t)\notin D.
    \end{cases}
\end{equation*}
Now we show that the section $s(k_1,k_2,t)$ vanishes only at the helix $\gamma$. Since
\begin{equation*}
\begin{split}
    \left|s_{\mathrm{local}}(k_1,k_2,t)-k^m\right|&=|k|^m\left|\left(1-\frac{\gamma_1(t)+i\gamma_2(t)}{k_1+ik_2}\right)^m-1\right|\\
    & \leq |k|^m\max\left((1-(1-a/c_0)^m,(1+a/c_0)^m-1\right),
\end{split}
\end{equation*}
by taking $a$ sufficiently small, we have
\begin{equation*}
    |s(k_1,k_2,t)|\geq |k|^m - |k|^m\max\left((1-(1-a/c_0)^m,(1+a/c_0)^m-1\right))>0,\quad c_0\leq |k|\leq 2c_0.
\end{equation*}
Convolving with a smooth approximation of identity near the gluing region gives a smooth section that only vanishes at the helix and has the form
\begin{equation*}
     s_{\mathrm{local}}(k_1,k_2,t):=(k_1+ik_2-\gamma_1(t)-i\gamma_2(t))^m
\end{equation*}
near the helix.
\end{proof}


Now using the section constructed in Proposition \ref{prop:vanish-t3}, we construct Wannier functions and compute the decay of the Wannier functions for non-trivial complex line bundles over $\TT^3$.

\begin{proof}[Proof of Theorem \ref{theo:becker}]

Similar to the proof of Theorem \ref{theo:thouless}, we first normalize the section constructed in Proposition \ref{prop:vanish-t3}. The normalized section is smooth everywhere away from the singularity at the helix curve $\gamma$ given by 
$$\frac{(k_1+ik_2-\gamma_1(t)-i\gamma_2(t))^m}{|k_1+ik_2-\gamma_1(t)-i\gamma_2(t)|^m}\Phi(k_1,k_2,t,x)$$
in a tubular neighborhood $D$ of $\gamma$, where $\Phi(k_1,k_2,t,x)$ is a local normalized smooth section. Let $\chi\in C_c^{\infty}([0,\infty))$ be a smooth cutoff function such that $\chi= 1$ near $0$. We again use \eqref{eq:Wannier} to recover the Wannier function from the Bloch function and reduce the integration over $\TT^2_{k_1,k_2}$ to the integration over $\mathbb R^2$, due to the presence of the cutoff function. The Wannier function is given by
\begin{equation}
\label{eq:wannier3}
\begin{split}
    w(\varphi)(x) = &
    \int_0^1\int_{\RR^2}e^{i(k_1x_1+k_2x_2+tx_3)} \frac{(k_1+ik_2-\gamma_1(t)-i\gamma_2(t))^m}{|k_1+ik_2-\gamma_1(t)-i\gamma_2(t)|^m} \\
    & \chi(|k_1+ik_2-\gamma_1(t)-i\gamma_2(t)|)\Phi(k_1,k_2,t,x)\, dk_1\,dk_2\, dt+\mathcal O(|x|^{-\infty})\\
    = &
    \int_0^1\int_{\RR^2}e^{i(k_1x_1+k_2x_2+tx_3)} \frac{(k_1+ik_2-\gamma_1(t)-i\gamma_2(t))^m}{|k_1+ik_2-\gamma_1(t)-i\gamma_2(t)|^m} \\
    & \chi(|k_1+ik_2-\gamma_1(t)-i\gamma_2(t)|)\Phi(\gamma(t),x)\, dk_1\,dk_2\, dt\\
    + &
    \int_0^1\int_{\RR^2}e^{i(k_1x_1+k_2x_2+tx_3)} \frac{(k_1+ik_2-\gamma_1(t)-i\gamma_2(t))^m}{|k_1+ik_2-\gamma_1(t)-i\gamma_2(t)|^m} \\
    & \chi(|k_1+ik_2-\gamma_1(t)-i\gamma_2(t)|)(\Phi(k_1,k_2,t,x)-\Phi(\gamma(t),x) )\, dk_1\,dk_2\, dt+\mathcal O(|x|^{-\infty})
\end{split}
\end{equation}

We consider the Fourier transform
\begin{equation*}
\begin{aligned}
    \mathcal{I}_{\gamma,\chi,\Phi}(x)=& \int_{0}^1\int_{\RR^2} \frac{(k_1+ik_2-\gamma_1(t)-i\gamma_2(t))^m}{|k_1+ik_2-\gamma_1(t)-i\gamma_2(t)|^m} e^{i(k_1 x_1+k_2 x_2+t x_3)} \\
    &\chi(|k_1+ik_2-\gamma_1(t)-i\gamma_2(t)|) \Phi(\gamma(t),x) dk_1dk_2dt\\
    = & \int_0^{\infty} \int_0^{2\pi}\int_0^1 e^{im\theta} e^{i( x_1(\gamma_1(t)+\cos\theta)+ x_2(\gamma_2(t)+\sin\theta)+t x_3)}\chi(r) \Phi(\gamma(t),x) rdrd\theta dt\\
    = & \int_0^{\infty} \int_0^{2\pi} e^{im\theta} e^{i(x_1\cos\theta+x_2\sin\theta))}\chi(r)rdrd\theta \cdot \int_0^1 e^{i(\gamma_1(t) x_1+\gamma_2(t) x_2+t x_3)}\Phi(\gamma(t),x)dt.
\end{aligned}
\end{equation*}
By \eqref{eq:computation}, the first term is estimated by
\begin{equation}
\label{eq:decay12}
    \int_0^{\infty} \int_0^{2\pi} e^{im\theta} e^{i(x_1\cos\theta+x_2\sin\theta))}\chi(r)rdrd\theta = \mathcal{O}((1+|x_1|+|x_2|)^{-2}). 
\end{equation}
Consider 
\begin{equation*}
    I_k(x)=\{t\in \RR/\ZZ:\frac{d^k}{dt^k}(\gamma_1(t)x_1+\gamma_2(t)x _2+tx_3)\neq 0\},\quad k=1,2,3,\quad x\neq 0.
\end{equation*}
Since $\gamma'(t),\gamma''(t),\gamma'''(t)$ are linearly independent, we have $I_1(x)\cup I_2(x)\cup I_3(x)=\RR/\ZZ$. Take a smooth partition of unity $1=\chi_1(t)+\chi_2(t)+\chi_3(t)$ on $\RR/\ZZ$ with $\supp \chi_k(t)\subset I_k(x)$. Note this partition of unity may depend on $x$, but we can choose it locally uniformly in terms of $\frac{x}{|x|}$.
By van der Corput lemma,
\begin{equation*}
    \int_0^1 e^{i(\gamma_1(t)x_1+\gamma_2(t)x _2+tx_3)} \chi_k(t)\Phi(\gamma(t),x)dt = \mathcal O(|x|^{-1/k}),\ k=1,2,3.
\end{equation*}
Moreover, when $|x_1|+|x_2|\ll |x_3|$, we have $\RR/\ZZ=I_1(x)$ and nonstationary phase gives
\begin{equation*}
    \int_0^1 e^{i(\gamma_1(t)x_1+\gamma_2(t)x _2+tx_3)} \chi_1(t)\Phi(\gamma(t),x)dt = \mathcal O(|x|^{-\infty}).
\end{equation*}
Therefore we conclude
\begin{equation*}
    \mathcal{I}_{\gamma,\chi,\Phi}(x) = \mathcal O(|x|^{-7/3}). 
\end{equation*}
For the second term in \eqref{eq:wannier3}, we have
\begin{equation*}
    \Phi(k_1,k_2,t,x)-\Phi(\gamma(t),x)=\sum\limits_{0<|\alpha|\leq N_0}\frac{(k_1-\gamma_1(t))^{\alpha_1}(k_2-\gamma_2(t))^{\alpha_2}}{\alpha!}\partial_k^{\alpha}\Phi(\gamma(t),x)+R_{N_0}(k_1,k_2,t,x)
\end{equation*}
where $\frac{(k_1+ik_2-\gamma_1(t)-i\gamma_2(t))^m}{|k_1+ik_2-\gamma_1(t)-i\gamma_2(t)|^m}R_{N_0}(k_1,k_2,t,x)\in C^{N_0}$.
By the similar estimates as before, we have
\begin{equation*}
    \begin{split}
    &\int_0^1\int_{\RR^2}e^{i(k_1x_1+k_2x_2+tx_3)} \frac{(k_1+ik_2-\gamma_1(t)-i\gamma_2(t))^m}{|k_1+ik_2-\gamma_1(t)-i\gamma_2(t)|^m} \\
    & \chi(|k_1+ik_2-\gamma_1(t)-i\gamma_2(t)|)(k_1-\gamma_1(t))^{\alpha_1}(k_2-\gamma_2(t))^{\alpha_2}\partial_k^{\alpha}\Phi(\gamma(t),x)\, dk_1\,dk_2\, dt\\
    =& \mathcal{O}(|x|^{-7/3-|\alpha|}).
    \end{split}
\end{equation*}
Moreover, 
\begin{equation*}
\begin{split}
&\int_0^1\int_{\RR^2}e^{i(k_1x_1+k_2x_2+tx_3)} \tfrac{(k_1+ik_2-\gamma_1(t)-i\gamma_2(t))^m}{|k_1+ik_2-\gamma_1(t)-i\gamma_2(t)|^m} R_{N_0}(k_1,k_2,t,x) \chi(|k_1+ik_2-\gamma_1(t)-i\gamma_2(t)|)\, dk_1\,dk_2\, dt \\
&= \mathcal{O}(|x|^{-N_0}). 
\end{split}
\end{equation*}
Taking $N_0=3$ finishes the proof.
\end{proof}

If we allow anisotropic decay, then we may also construct a section of the complex line bundle $\mathcal{E}$ vanishing along $\gamma(t) = (0,0,t), t\in [0,1]$ such that 
\begin{equation*}
s(k_1,k_2,t)=(k_1+ik_2)^m
\end{equation*}
near the curve $\gamma(t)$. By separation of variables, proof of Theorem \ref{theo:thouless}, and non-stationary phase, we can conclude the following

\begin{corr}
\label{cor:Becker}
There exists a Wannier function $w(\varphi)$ for a complex Bloch line bundle which exhibits $\mathcal{O}((1+|x_1|+|x_2|)^{-2}\langle x_3\rangle^{-\infty})$ decay.
\end{corr}

\end{document}